\documentclass[12pt]{iopart}
\usepackage{amsthm,amssymb,nath}
\usepackage{graphicx}

\let\epsilon\varepsilon

\def\eqref#1{{\rm(\ref{#1})}}
\newtheorem{proposition}{Proposition}
\newtheorem{corollary}{Corollary}
\newtheorem{definition}{Definition}

\newtheorem{remark}{Remark}
\newtheorem{example}{Example}
\newtheorem{theorem}{Theorem}

\def\lat{\theta}
\def\lon{\phi}

\def\Involution{\mathfrak I}
\def\Scaling{\mathfrak S}
\def\xTranslation{\mathfrak T^x}
\def\yTranslation{\mathfrak T^y}
\def\scaling{\mathfrak s}
\def\xtranslation{\mathfrak t^x}
\def\ytranslation{\mathfrak t^y}
\def\dif{\mathop{}\!\mathrm d}

\begin{document}


\title{Some results concerning the constant astigmatism equation}
\author{Adam Hlav\'a\v{c} and Michal Marvan}
\address{Mathematical Institute in Opava, Silesian University in
  Opava, Na Rybn\'\i\v{c}ku 1, 746 01 Opava, Czech Republic.
  {\it E-mail}: Adam.Hlavac@math.slu.cz, Michal.Marvan@math.slu.cz}
\date{}

\ams{35Q53, 53A05, 58J72, 73E15}


\begin{abstract}
In this paper we continue investigation of the constant astigmatism equation
$z_{yy} + (1/z)_{xx} + 2 = 0$. We newly interpret its solutions as 
describing spherical orthogonal equiareal patterns, with relevance to 
two-dimensional plasticity. 
We show how the classical Bianchi superposition principle for the sine-Gordon 
equation can be extended to generate an arbitrary number of solutions of the 
constant astigmatism equation by algebraic manipulations. 
As a by-product, we show that sine-Gordon solutions give slip line fields on the 
sphere. 
Finally, we compute the solutions corresponding to classical Lipschitz 
surfaces of constant astigmatism via the corresponding equiareal patterns.
\end{abstract}

\section{Introduction}

It is well known that the classical B\"acklund 
transformation~\cite{Bae} for the sine-Gordon equation $u_{\xi\eta} = \sin u$ 
as well as the Bianchi permutability property~\cite{Bia b}
have been discovered in the context of pseudospherical surfaces, 
i.e., surfaces of constant negative Gaussian curvature.
It is perhaps less known that historical roots of these developments 
lie in another class of surfaces, characterised by the constancy of the 
difference $\rho_2 - \rho_1$ between the principal radii of curvature 
$\rho_1,\rho_2$; see~\cite{P-S} for the historical account.
Lying covered with dust and oblivion for almost a century, the surfaces 
satisfying $\rho_2 - \rho_1 = `const$ reemerged recently from the systematic 
search for integrable classes of Weingarten surfaces conducted by Baran and 
one of us~\cite{B-M I}. 
Although nameless in the nineteenth century, in~\cite{B-M I} they have been
named the {\it surfaces of constant astigmatism} in connotation with 
the astigmatic interval~\cite{Sturm} of the geometric optics, albeit without 
suggesting any specific application.

Undoubtedly, the most important results about constant astigmatism surfaces are 
due to Bianchi.
In~\cite{Bia a} (see also~\cite[\S130]{Bia I}), Bianchi observed that 
evolutes (i.e., focal surfaces) of surfaces satisfying 
$\rho_2 - \rho_1 = `const$ are pseudospherical. 
In the same paper he also constructed surfaces satisfying 
$\rho_2 - \rho_1 = `const$ as involutes corresponding to parabolic 
geodesic systems on pseudospherical surfaces.
Apparently, Bianchi was the first to obtain surfaces of constant 
astigmatism explicitly, namely, surfaces~\cite[eq.~(30)]{Bia a} 
corresponding to Dini's pseudospherical helicoids 
(see, e.g.,~\cite[\S1.4.2]{R-S} or~\cite[p.~183]{Sym}). 
Lipschitz~\cite{Lip} obtained another class of surfaces of constant astigmatism; 
within the full class given in terms of elliptic integrals he pointed out a 
subclass of surfaces of revolution, further investigated by von Lilienthal. 

Let us stress that the aforementioned constructions of Lipschitz and Bianchi 
refer to ad hoc parameterisations. 
Bianchi used the rotation angle and a parameterisation of 
the generating tractrix of the helicoid, while Lipschitz employed spherical 
coordinates on the Gaussian sphere. 
In~\cite{B-M I} we observed that under an adapted parameterisation by lines 
of curvature the constant astigmatism surfaces correspond to solutions 
of the {\it constant astigmatism equation}
$$
\numbered\label{CAE}
z_{yy} + (\frac1z)_{xx} + 2 = 0
$$
($x,y$ are natural parameters in the sense of Ganchev and Mihova~\cite{G-M}).
The geometric link to pseudospherical surfaces induces a nonlocal 
transformation to the sine-Gordon equation and vice versa. Since curvature 
coordinates on constant astigmatism surfaces correspond to parabolic geodesic 
coordinates on the pseudo\-spherical surfaces, and these are not the coordinates 
the sine-Gordon equation is referred to, the transformations  change both 
the dependent and independent variables.
Explicit formulas can be found in~\cite{B-M I}, ready to be applied to 
the sine-Gordon solutions, which are known in abundance, 
see~\cite{A-D-M,D-N,Ov,P-M} and references therein. 
However, the only explicit instance of such a relationship we were able to find 
in the literature was that of the Bianchi surfaces~\cite[eq.~(30)]{Bia a} to
the Dini helicoid; Fig.~\ref{fig:invDini} presents a plot of them as 
unparameterised surfaces.
\begin{figure}[h]
\begin{center}
\includegraphics[scale=0.25,angle=270]{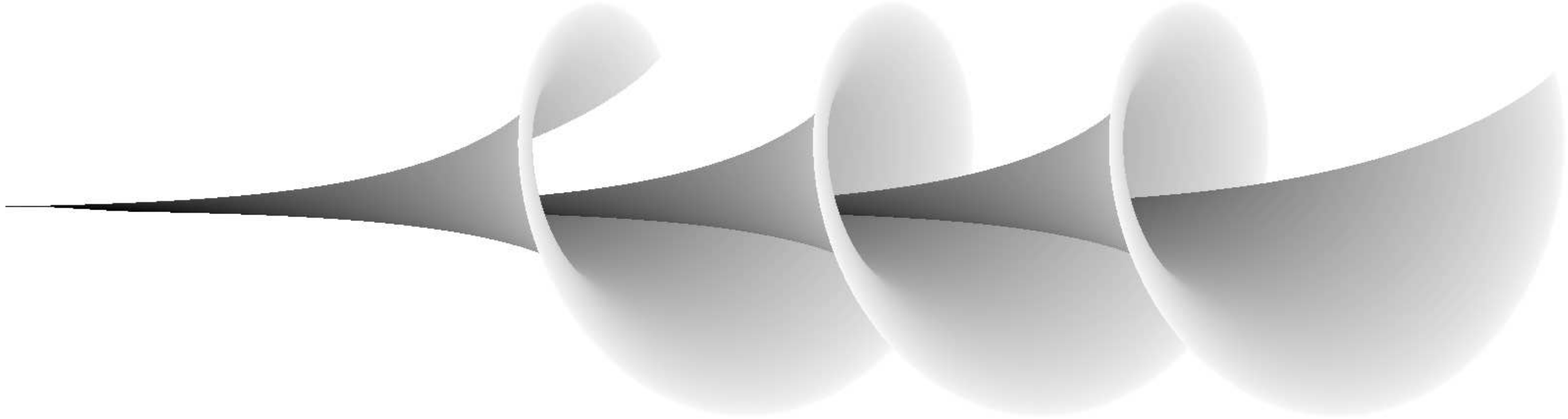}
\includegraphics[scale=0.25,angle=270]{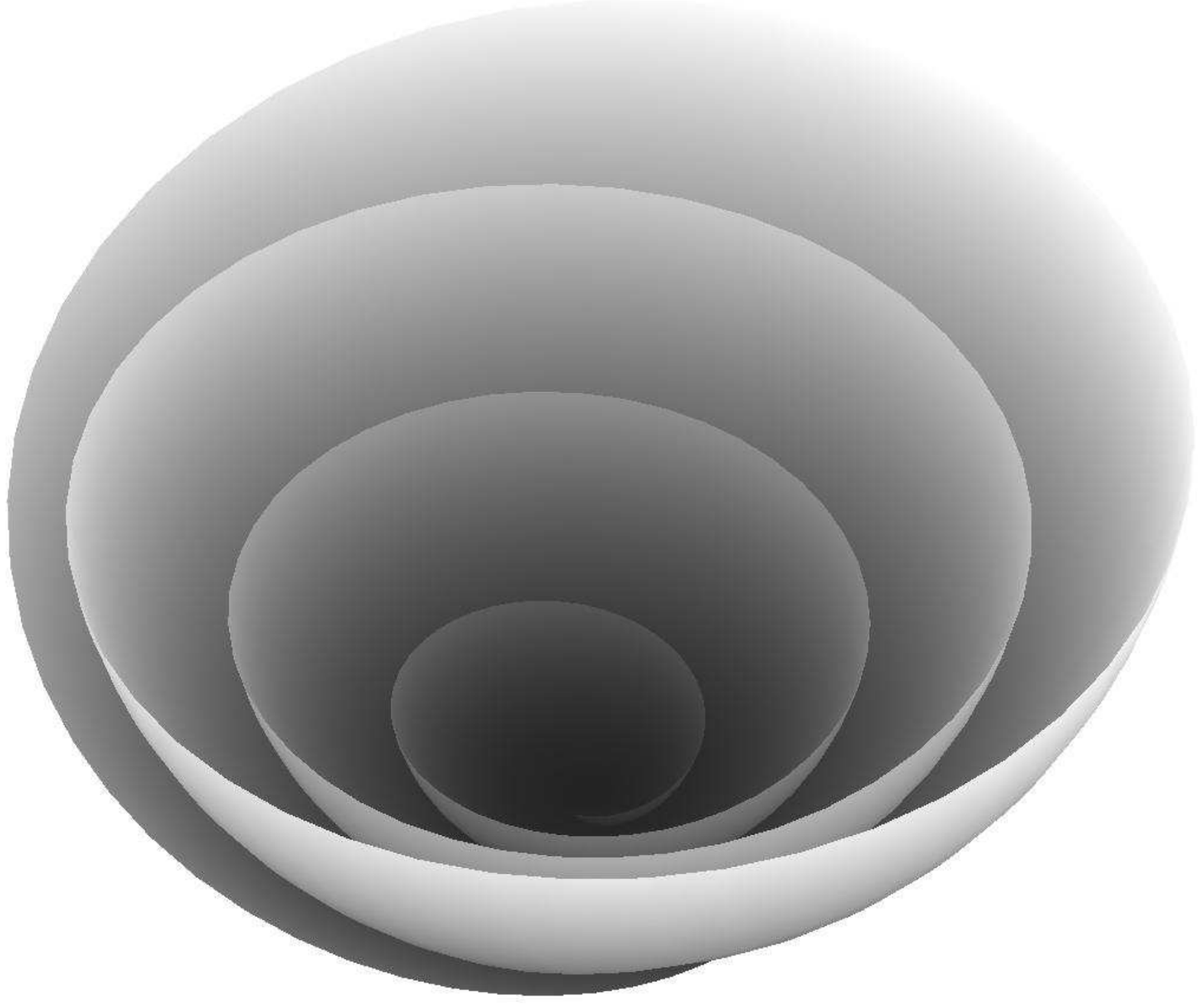}
\caption{Dini's pseudospherical surface (left) and its constant astigmatism
involute (right)}
\label{fig:invDini}
\end{center}
\end{figure}

In this paper we continue the investigation of the constant astigmatism 
equation~\eqref{CAE} and its solutions.
Firstly, we show that equation~\eqref{CAE} describes 
{\it orthogonal equiareal patterns} (Sadowsky~\cite{Sad I,Sad II}) on the sphere, 
i.e., a system of local coordinates $x^1,x^2$ such that the metric coefficients
satisfy $g_{12} = 0$, $\det g = 1$.
Hence, the area element is simply $\dif x^1 \wedge \dif x^2$ and the area 
of the curvilinear rectangle $a^i \le x^i \le b^i$, $i = 1,2$, is equal to
$(b^1 - a^1)(b^2 - a^2)$.
It follows that the curvilinear rectangles formed by ``uniformly spaced'' 
coordinate lines are of equal area, which explains the terminology.

The equiareal property in plasticity theory can be traced back to 
Boussinesq~\cite{Bou}. Seventy years later, Sadowsky~\cite{Sad I,Sad II}, 
rediscovered the ``equiareal patterns'' as configurations of the principal 
stress lines under the Tresca yield condition (see~\cite{Hill book}) and 
gave them their name.
Later Hill~\cite{Hill} gave a kinematic interpretation of these patterns.
Coburn~\cite[Thm.~1]{Cob} established the same equiareal property, this time for 
slip lines under a different yield condition. 
Ament~\cite{Ame} discovered a relation to the class of Weingarten surfaces, 
determined by the constancy of the difference between the principal curvatures 
(as opposed to the difference between the principal radii of curvature). 
Finally, Fialkow~\cite[Th.~4.1]{Fia} observed relevance of orthogonal equiareal 
patterns to conformal geometry. 

The contents of this paper are as follows. Section~2 contains the necessary background.
In Section~\ref{sect:oep}, we observe that adapted curvature coordinates on 
constant astigmatism surfaces correspond to orthogonal equiareal 
patterns on the Gaussian sphere. Inspired by the aforementioned relation to 
plasticity, we construct a two-dimensional stress tensor formally satisfying 
both the Tresca yield condition and the equilibrium equations.
Physical relevance of our purely mathematical construction is not a primary concern, 
yet the flow of a thin plastic layer around a sphere seems to be a realistic picture.
Guided by this picture we investigate the maximum shear stress directions, 
positioned at the angle of $\pi/4$ to the principal stress directions.
The corresponding trajectories are known as slip lines; we show them to 
be related to solutions of the sine-Gordon equation later in section~\ref{sect:BTSP}.

Section~\ref{sect:BTSP} is devoted to a simplified reconstruction of constant 
astigmatism surfaces from a pair of complementary pseudospherical 
surfaces~\cite{Bia a} or~\cite[\S136]{Bia I}, 
i.e., under the frequently occurring condition that both evolutes are known. 
Complementary pseudospherical surfaces are easy to find 
among those resulting from the famous and powerful Bianchi permutability 
theorem~\cite{Bia b}. 
It turns out that given a pair of complementary pseudospherical surfaces, the
corresponding (unparameterised) constant astigmatism surface can be obtained 
by pure algebraic manipulations and differentiation. 
This is also true for geodesics (as proved by Bianchi himself) and, hence, 
for one of the curvature coordinates, while obtaining the other requires one 
integration. However, owing to a suitable extension of the Bianchi superposition 
principle this integration needs to be done only once. 
We also observe (Proposition~\ref{prop:y}) that the coordinates 
$\xi,\eta$ the sine-Gordon equation is referred to correspond to slip line 
fields on the spherical image of the constant astigmatism surface.

Finally, in section~\ref{sect:Lipsch} we pay another longstanding debt and 
find the function $z(x,y)$ corresponding to the Lipschitz surfaces.
As we already mentioned, Lipschitz computed a class of constant astigmatism 
surfaces in terms of spherical coordinates on the Gaussian image. 
The result being not easily transformable to curvature coordinates, we compute 
the associated orthogonal equiareal pattern directly from the definitions
to observe that solutions of the Lipschitz class are invariant solutions with 
respect to Lie symmetries.

\section{Preliminaries}
\label{sect:prelim}

In this section we recall previous results about the constant astigmatism 
surfaces; see~\cite{B-M I} for details.
We consider surfaces immersed in Euclidean space under parameterisation by 
the lines of curvature (also known as curvature coordinates).  
Hence, the fundamental forms can be written as
$$
\begin{gathered}
\mathbf{I} = u^2 \dif x^2 + v^2 \dif y^2\,, \quad
\mathbf{II} = \frac{u^2}{\rho_1} \dif x^2 + \frac{v^2}{\rho_2} \dif y^2\,, \quad
\mathbf{III} = \frac{u^2}{\rho_1^2} \dif x^2 + \frac{v^2}{\rho_2^2} \dif y^2\,,
\end{gathered} 
$$
where $\rho_1,\rho_2$ are the principal radii of curvature.
The first two forms determine the surface up to the rigid motions (Bonnet theorem).

A surface is said to be of {\it constant astigmatism} 
if the difference $\rho_2 - \rho_1$ between the principal radii of curvature 
is a nonzero constant (if zero, then the surface is a part of the sphere).
We assume the ambient space to be scaled so that $\rho_2 - \rho_1 = \pm1$. 

\begin{definition} \rm
A parameterisation by lines of curvature is said to be {\it adapted\/} if 
$$
\numbered\label{adapted}
u v (\frac1{\rho_1} - \frac1{\rho_2}) = \pm1
$$ 
holds.
\end{definition}

This is the natural parameterisation recently introduced by Ganchev and 
Mihova~\cite[Prop.~5.6]{G-M} with the arbitrary constant being normalised to $\pm1$.
Every constant astigmatism (more generally, Weingarten) surface can be equipped 
with an adapted parameterisation by lines 
of curvature, see~\cite[Prop.~5.6]{G-M} or~\cite{B-M I}. 
Henceforth we assume that $x,y$ are adapted coordinates.
Then, according to~\cite{B-M I}, the nonzero coefficients of the three fundamental 
forms of a surface of constant astigmatism can be expressed through a 
single variable $z(x,y)$:
$$
u = \frac{z^{\frac{1}{2}}(\ln z-2)}{2} , \qquad 
v = \frac{\ln z}{2 z^{\frac{1}{2}}}, \qquad 
\rho_1 = \frac{\ln z-2}{2}, \qquad 
\rho_2 = \frac{\ln z}{2}.  
$$
Obviously, condition~\eqref{adapted} is satisfied. 

Let $\mathbf{r}(x,y)$ be the surface of constant astigmatism corresponding to 
$z(x,y)$, let $\mathbf n(x,y)$ denote the unit normal vector. 
Then $\mathbf{r}, \mathbf n$ satisfy the Gauss--Weingarten system
$$
\numbered\label{GW:CA}
\mathbf{r}_{xx} = \frac{(\ln z) z_x}{2 (\ln z - 2) z} \mathbf{r}_x
                - \frac{(\ln z - 2) z z_y}{2 \ln z} \mathbf{r}_y
                + \frac12 (\ln z - 2) z \mathbf{n},
\\
\mathbf{r}_{xy} = \frac{(\ln z) z_y}{2 (\ln z - 2) z} \mathbf{r}_x
                - \frac{(\ln z - 2) z z_x}{2 \ln z} \mathbf{r}_y,
\\
\mathbf{r}_{yy} = \frac{(\ln z) z_x}{2 (\ln z - 2) z^3} \mathbf{r}_x
                - \frac{(\ln z - 2) z_y}{2 z \ln z} \mathbf{r}_y
                + \frac{\ln z}{2 z} \mathbf{n},
\\
\mathbf{n}_x = -\frac{2}{\ln z - 2} \mathbf{r}_x, 
\qquad
\mathbf{n}_y = -\frac{2}{\ln z} \mathbf{r}_y.
$$
Note that
$\mathbf{e}_1 = \frac{\mathbf{r}_x}{u}$,
$\mathbf{e}_2 = \frac{\mathbf{r}_y}{v}$, and
$\mathbf{n} = \mathbf{e}_1 \times \mathbf{e}_2$
constitute an orthonormal frame. 

Compatibility conditions of the Gauss--Weingarten system constitute 
the Gauss--Mainardi--Codazzi system, which in our case amounts to the 
Gauss equation alone, and coincides with the 
{\it constant astigmatism equation}~\eqref{CAE}.

According to Bianchi~\cite{Bia a} (see also~\cite[\S136]{Bia II}), 
if $\mathbf r$ is a surface of constant astigmatism and $\mathbf n$ is 
its normal, then the two evolutes 
$$
\numbered\label{comp evolutes}
\mathbf r + \rho_1 \mathbf n, \quad \mathbf r + \rho_2 \mathbf n
$$
are pseudospherical surfaces. These are said to be~{\it complementary}.

For further reference, we also recall a list of symmetries of equation~\eqref{CAE}.
Lie symmetries are completely known, see~\cite{B-M I}. They are
the $x$-{\it translation} 
$\xTranslation_c (x, y, z) = (x + c, y, z)$,  
the $y$-{\it translation}
$\yTranslation_c (x, y, z) = (x, y + c, z)$, 
and the {\it scaling}
$\Scaling_c (x, y, z) = (`e^{-c} x, `e^c y, `e^{2c} z)$, where $c$ is a real parameter. 

We shall also refer to a discrete symmetry 
$\Involution(x, y, z) = (y, x, \frac 1 z)$, called the {\it involution}.
Obviously,
$$
\begin{gathered}
\Involution \circ \Involution = `Id, \\
\Involution \circ \xTranslation_a = \yTranslation_{a} \circ \Involution, &
\Involution \circ \yTranslation_a = \xTranslation_{a} \circ \Involution, \\
\Scaling_c \circ \xTranslation_a = \xTranslation_{a/c} \circ \Scaling_c, &
\Scaling_c \circ \yTranslation_b = \yTranslation_{cb} \circ \Scaling_c, \\
\Scaling_c \circ \Involution = \Involution \circ \Scaling_{1/c}.
\end{gathered}
$$
Translations are mere reparameterisations of the corresponding 
constant astigmatism surface. The scaling symmetry corresponds to an 
{\it offsetting}, i.e., takes a surface to a parallel surface 
(moves every point a unit distance along the normal). The involution 
interchanges $x$ and $y$ (swaps the orientation), followed by a unit offsetting.

\section{Orthogonal equiareal patterns and slip line fields}
\label{sect:oep}

The geometric meaning of the variable $z$ can be seen from the third
fundamental form, which turns out to be simply
$$
\mathbf{III} = z\dif x^2 + \frac 1z \dif y^2. 
$$
Since $\mathbf{III} = \dif\mathbf{n} \cdot \dif\mathbf{n}$ coincides with 
the first fundamental form of the Gaussian sphere $\mathbf{n}(x,y)$, 
it follows that one obtains a rather special parameterisation of the latter.

\begin{definition} \rm
By an {\it orthogonal equiareal pattern} on a surface $S$ we shall 
mean a parameterization $x,y$ such that the corresponding first fundamental 
form is
$$
\numbered\label{oep}
\mathbf{I}_S = z\dif x^2 + \frac 1z \dif y^2, 
$$
$z$ being an arbitrary function of $x,y$.
\end{definition}

Let $\mathbf R$ denote the position vector of a point 
on the surface~$S$. Since $\det \mathbf{I}_S = 1$, the local parameterisation 
$\mathbf{R}(x,y)$ is an area preserving map from the plane to the 
surface $S$. Moreover, the coordinate lines are, obviously, orthogonal.
These two properties imply that evenly distributed coordinate lines 
cover the surface with curvilinear rectangles of equal area (see the 
Introduction).

\begin{example} \rm {\it The Archimedean projection.} \label{ex:Arch}
A simple example of an orthogonal equiareal pattern on the
sphere that can be seen on Fig.~\ref{fig:Archimed} is delivered by the 
well-known Archimedean projection  
of the cylinder $(\cos y, \sin y, x)$ onto an inscribed sphere.
In this case, $(x,y)$ is sent to $(\sqrt{1 - x^2} \cos y, \sqrt{1 - x^2} \sin y, x)$
and we have
$$
\mathbf{I}_{`Arch} = \frac{\dif x^2}{1 - x^2} + (1 - x^2) \dif y^2, 
$$
i.e., $z = 1/(1 - x^2)$.
According to~\cite{B-M I}, this solution of the constant 
astigmatism equation corresponds to von Lilienthal surfaces~\cite{vLil}.
\end{example}

\begin{figure}[h]
\begin{center}
\includegraphics[scale=0.35]{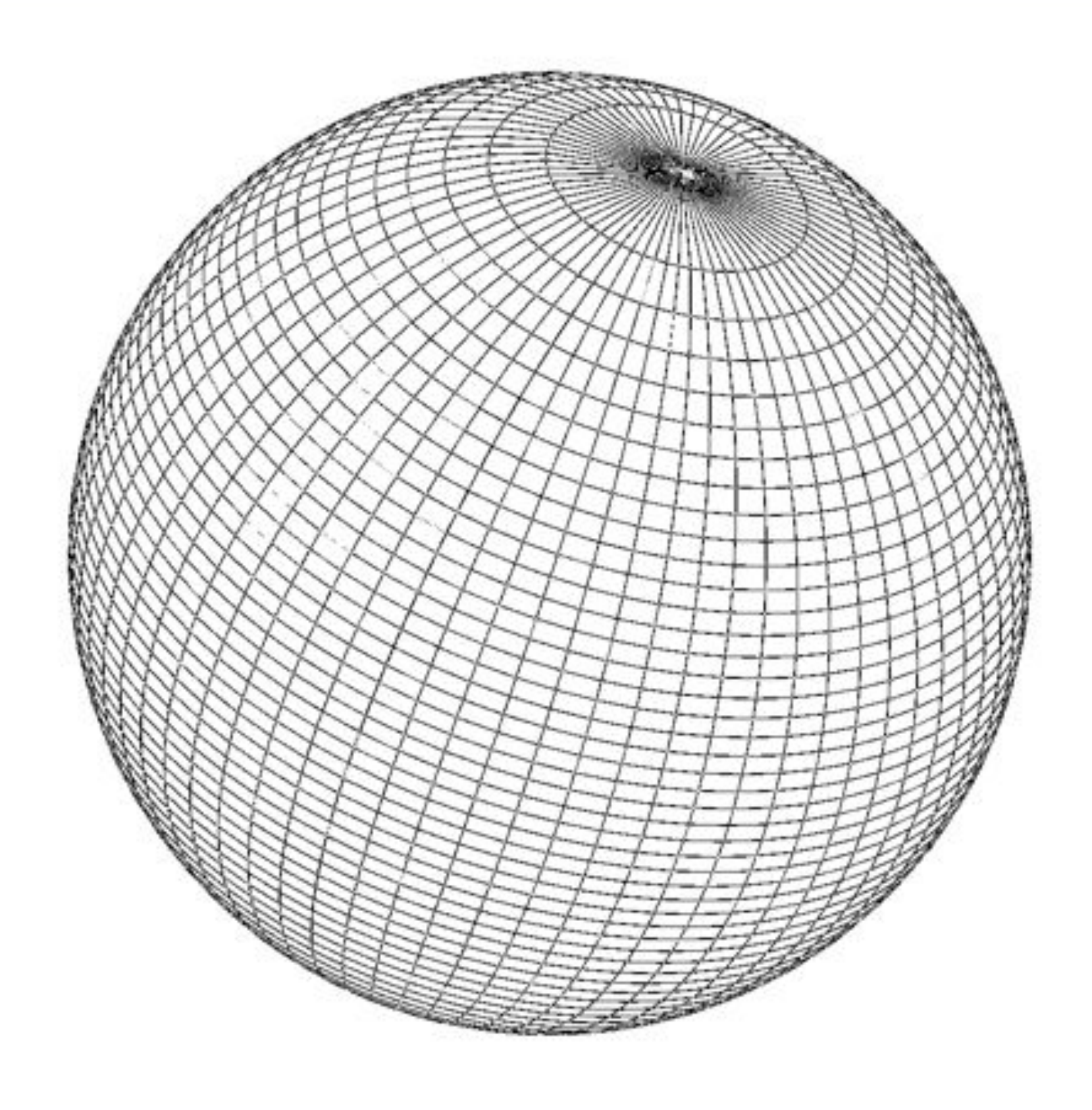}
\caption{The Archimedean equiareal parameterisation of the sphere}
\label{fig:Archimed}
\end{center}
\end{figure}

Not only every constant astigmatism surface generates an orthogonal equiareal 
parameterization of the unit sphere; a converse statement is also available.

\begin{proposition} 
Let $\mathbf n(x,y)$, $\|\mathbf n\| = 1$, be an orthogonal equiareal pattern
on the unit sphere $S$.
Then $z$ defined by formula~\eqref{oep} is a solution of the constant 
astigmatism equation~\eqref{CAE}.
\end{proposition}

\begin{proof}
Using the well-known Brioschi formula to compute the Gaussian curvature of the
sphere, we obtain
$$
1 = -\frac12z_{yy} - \frac12(\frac1z)_{xx}.
$$
The constant astigmatism equation~\eqref{CAE} easily follows. 
\end{proof}

The corresponding constant astigmatism surface can be reconstructed from
the last two equations of the Gauss--Weingarten system~\eqref{GW:CA}.

Let us stress that all the point symmetries given in Sect.~\ref{sect:prelim}
can be understood as reparameterisations of the corresponding orthogonal 
equiareal pattern on the Gaussian sphere.
In particular, scaling~$\Scaling_c$ means shrinking the pattern along one 
family of lines, compensated by stretching it along the orthogonal 
family of lines.

In the case of $S$ being a plane, the notion of an orthogonal equiareal 
pattern was introduced by Sadowski~\cite{Sad I,Sad II} in the context of 
two-dimensional plasticity. Choosing the vectors $\partial_x, \partial_y$ along the 
{\it principal stress directions} (i.e., eigenvectors of the 
stress tensor $\sigma^i_j$), Sadowski derived the equiareal property from the 
equilibrium condition $`div \sigma = 0$ and the Tresca yield condition
$\sigma^1_1 - \sigma^2_2 = `const$. 


Let us reverse the line of reasoning and reconstruct
a two-dimensional stress tensor from a given orthogonal equiareal 
pattern $g = `I_S$.
In what follows, all components are taken with respect to the basis 
$\partial_x, \partial_y$ of the tangent space and
indices are raised and lowered with the metric.

\begin{proposition}
Consider an orthogonal equiareal pattern $g = g_{ij} \dif x^i \dif x^j$ such that
$$
g_{11} = z,
 \quad g_{12} = g_{21} = 0,
 \quad g_{22} =  1/z.
$$ 
Then the tensor $\sigma$ given by the components
$$
\numbered\label{stress ud}
\sigma^1_1 = \frac12 \ln z,
 \quad \sigma^1_2 = \sigma^2_1 = 0,
 \quad \sigma^2_2 = \frac12 (\ln z-2).
$$ 
satisfies
$\sigma^{ij}_{;j} = 0$ (the equilibrium equation)
and $\sigma^1_1 - \sigma^2_2 = 1$ (the Tresca yield condition). 
\end{proposition}

\begin{proof}
From the metric coefficients we produce the Christoffel symbols
$$
\Gamma^1_{11} = -\Gamma^2_{12} = \frac{z_x}{2z}, \quad
\Gamma^1_{12} = -\Gamma^2_{22} = \frac{z_y}{2z}, \quad
\Gamma^1_{22} = \frac{z_x}{2z^3}, \quad
\Gamma^2_{11} = -\frac{z z_y}{2}
$$
as well as the twice contravariant tensor
$$
\numbered\label{stress uu}
 \sigma^{11} = \frac{\ln z}{2 z},
 \quad \sigma^{12} = \sigma^{21} = 0,
 \quad \sigma^{22} = \frac{\ln z - 2} {2} z.
$$ 
The yield condition $\sigma^1_1 - \sigma^2_2 = 1$ is obvious, while checking 
the equilibrium equation $\sigma^{ij}_{;j} = 0$ is a matter of routine.
\end{proof}

This proposition holds for any surface $S$ equipped with an
orthogonal equiareal pattern $g$. 
Note that if $S$ is a unit sphere, then the Tresca yield condition follows 
from the constant astigmatism property, since $\sigma^1_1 = \rho_2$ and 
$\sigma^2_2 = \rho_1$. Conversely, if $\sigma^1_1$ and $\sigma^2_2$ are 
arbitrary functions of $z$, then the equilibrium equation and 
the yield condition imply the same $\sigma^1_1$ and $\sigma^2_2$ as 
in~\eqref{stress ud} up to an additive and a multiplicative constant.


In the rest of this section we recall the derivation of the Mohr circle 
and slip lines in two-dimensional plasticity (see, e.g., \cite{Hill book}).
We consider a symmetric stress tensor $\sigma$ diagonalised along the principal stress directions 
$\partial_x, \partial_y$, i.e., 
$$
\sigma_1^1 = p, \quad \sigma_1^2 = \sigma_2^1 = 0, \quad \sigma_2^2 = q, \quad g_{12} = 0.
$$
Let $\mathbf w$ be a unit vector with components 
$(\cos\phi/\sqrt{g_{11}},\sin\phi/\sqrt{g_{22}})$ with 
respect to the basis $\partial_x,\partial_y$. 
The stress $\sigma(\mathbf w) = p \cos\phi/\sqrt{g_{11}} + q \sin\phi/\sqrt{g_{22}}$ 
can be decomposed into a sum of the normal stress
$\sigma_N(\mathbf w)$ and the shear stress $\sigma_T(\mathbf w)$, where 
$\sigma_N(\mathbf w) \mathrel{\bot} \sigma_T(\mathbf w)$ and, by definition, 
$\sigma_N(\mathbf w) = \alpha \mathbf w$ is a multiple of $\mathbf w$.
Hence,
$$
\sigma_T(\mathbf w) = \sigma(\mathbf w) - \alpha \mathbf w
 = ((p - \alpha) \cos\phi/\sqrt{g_{11}}, (q - \alpha) \sin\phi/\sqrt{g_{22}}).
$$
Obviously, $\|\sigma_N(\mathbf w)\| = |\alpha|$; likewise, we introduce
$\beta = \|\sigma_T(\mathbf w)\|$.
By orthogonality,
$$
0 = \sigma_N(\mathbf w) \cdot \sigma_T(\mathbf w)
 = \alpha (p - \alpha) \cos^2 \phi + \alpha (q - \alpha) \sin^2 \phi.
$$ 
Excluding $\cos\phi$ and $\sin\phi$ from the last equation and the condition 
$$
0 = \beta^2 - \sigma_T(\mathbf w) \cdot \sigma_T(\mathbf w)
 = (p - \alpha)^2 \cos^2 \phi + (q - \alpha)^2 \sin^2 \phi,
$$
we conclude that all admissible values of $\alpha,\beta$ belong to the 
{\it Mohr circle}~\cite{Moh} 
$$
(\alpha - \frac{p + q} 2)^2 + \beta^2 = (\frac{p - q} 2)^2.
$$
It follows that the extremal values $\alpha = p,q$ of the normal stress magnitude 
$\alpha$ are achieved when $\beta = 0$, i.e., when
$\mathbf w$ lies in one of the principal stress directions, as it should be.
The extremal values $|\frac12 (p - q)|$ of the shear stress magnitude 
$\beta$ are achieved when $\alpha = \frac12(p + q)$.
To find the corresponding vectors $\mathbf w$, we determine the acute
angle $\phi$ between $\mathbf w$ and $\pm\partial_x = (\pm 1,0)$.  
Substituting $\alpha = \frac12(p + q)$ into 
$\cos^2 \phi = (\alpha - q)/(p - q)$, 
we obtain $\cos^2 \phi = \frac12$, meaning that $\phi = \frac14\pi$.

Now, the Tresca criterion (see, e.g.,~\cite{Hill book})
says that yielding occurs whenever the maximal shear stress magnitude $\beta$ 
achieves a threshold depending on the material. It follows that the stress
tensor satisfies $p - q = `const$, which is called the Tresca yield condition.
The lines along the maximal shear stress direction are called slip lines and,
as we have already seen, have a constant deviation of $\frac14\pi$ from the
principal stress directions.

\begin{definition} \label{def:slip line} \rm
By a {\it slip line field} associated with the 
orthogonal equiareal pattern~\eqref{oep} on a surface $S$ we shall 
mean a parameterization $\xi,\eta$ such that the angle between
$\partial_x$ and $\partial_\xi$ as well as the angle between 
$\partial_y$ and $\partial_\eta$ is equal to $\frac14 \pi$.
\end{definition}

It follows that slip lines form an orthogonal net.
Note the available freedom of reparameterisation of each $\xi$ or $\eta$ separately.

\begin{remark} \rm
In plane plasticity slip line lines form what is called a 
{\it Hencky net}~\cite{Hen}. These have been fully described by
Carath\'eodory and Schmidt~\cite{C-S};
a full description of orthogonal equiareal patterns 
in the plane follows. 
However, the famous Hencky conditions satisfied by slip lines in the plane
fail on surfaces of non-vanishing curvature.
\end{remark}

\begin{example} \rm Continuing Example~\ref{ex:Arch}, we
easily see that the corresponding orthogonal net of slip lines is, by definition, 
formed by the $\pm 45^\circ$ 
loxodromes (lines of constant bearing); see Fig.~\ref{fig:ArchLoxo}
or model No. 249 
in the G\"ottingen collection~\cite{Gsammlung}.
Also compare with Zelin's superplastic sheet stretched with 
a spherical punch~\cite[Fig.~5b]{Zel}.
\end{example}

\begin{figure}
\begin{center}
\includegraphics[scale=0.34]{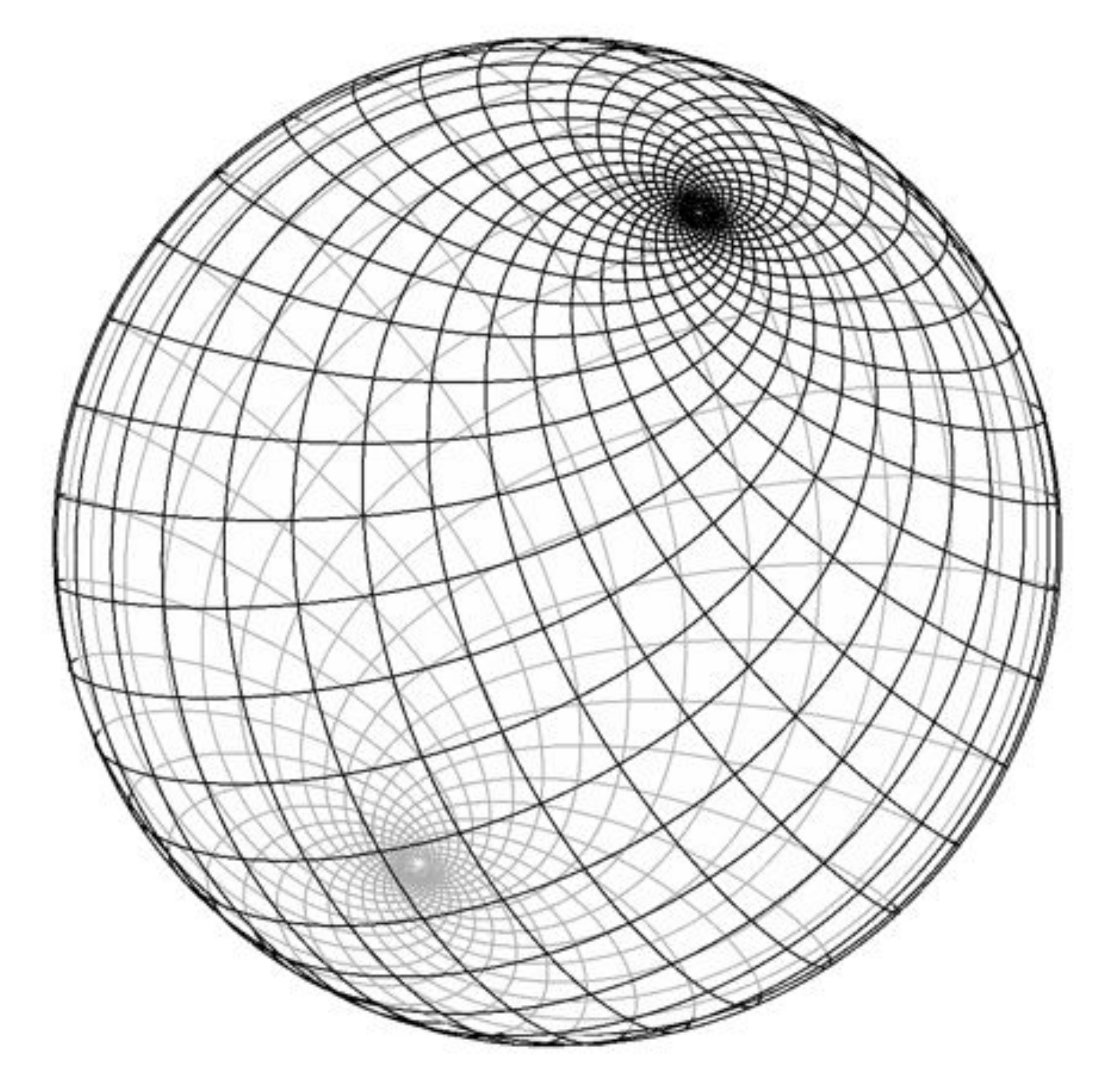}
\caption{Sphere's slip line field composed of loxodromes}
\label{fig:ArchLoxo}
\end{center}
\end{figure}

In the next section (Prop.~\ref{prop:y}) we shall see that solutions of 
the sine--Gordon equation produce slip line fields on the Gaussian sphere 
of the associated constant astigmatism surface.

\section{B\"acklund transformation and the superposition principle}
\label{sect:BTSP}

Available already at the beginning of the nineteenth century, the classical 
B\"acklund transformation~\cite{Bae} in combination with the Bianchi permutability 
theorem~\cite{Bia b} is a powerful way to generate pseudospherical 
surfaces and solutions of the sine--Gordon equation.
In this section we extend these methods to generate
constant astigmatism surfaces, solutions of the constant astigmatism equation 
as well as equiareal patterns and slip line fields on the sphere.

To start with we briefly recall the B\"acklund transformation~\cite{Bae}, see, 
e.g.,~\cite{Bia II,R-S,Sym}.
Let us consider a pseudospherical surface $\mathbf r(\xi,\eta)$, where
the parameters $\xi,\eta$ are both Chebyshev and asymptotic (which is always possible), 
i.e., 
$$
\mathbf I = \dif\xi^2 + 2\cos(2\omega)\dif\xi\dif\eta + \dif\eta^2, \quad
\mathbf{II} = 2\sin(2\omega)\dif\xi\dif\eta.
$$ 
The position vector $\mathbf r(\xi,\eta)$ and the unit normal $\mathbf n(\xi,\eta)$
satisfy the Gauss--Weingarten system
$$
\numbered\label{GW:SG}
\mathbf r_{\xi\xi} = \frac{\sin 4\omega}{\sin^2 2\omega} \omega_x \mathbf r_\xi
 - \frac{2}{\sin 2 w} \omega_x \mathbf r_\eta, \\
\mathbf r_{\xi\eta} = (\sin 2\omega) \mathbf n, \\
\mathbf r_{\eta\eta} = \frac{\sin 4\omega}{\sin^2 2\omega} \omega_y \mathbf r_\eta
 - \frac{2}{\sin 2 w} \omega_y \mathbf r_\xi, \\
\mathbf n_\xi = \frac{\sin 4\omega}{2 \sin^2 2\omega} \mathbf r_\xi
 - \frac{1}{\sin 2 w} \omega_x \mathbf r_\eta, \\
\mathbf n_\eta = \frac{\sin 4\omega}{2 \sin^2 2\omega} \mathbf r_\eta
 - \frac{1}{\sin 2 w} \omega_x \mathbf r_\xi. 
$$
The integrability conditions of the above system reduce to the sine-Gordon equation
$$
\omega_{\xi\eta} = \frac12 \sin 2\omega.
$$
The B\"acklund transform of our surface is
$$
\numbered\label{BTr}
\mathbf r^{(\lambda)} = \mathbf r
 + \frac{2\lambda}{1 + \lambda^2} (\frac{\sin(\omega - \omega^{(\lambda)})}{\sin(2\omega)} \mathbf r_\xi
 +             \frac{\sin(\omega + \omega^{(\lambda)})}{\sin(2\omega)} \mathbf r_\eta),
$$
where $\omega^{(\lambda)}$ is another sine-Gordon solution,
obtained from the pair of compatible first-order equations
$$
\numbered\label{BT}
\omega^{(\lambda)}_\xi = \omega_\xi + \lambda
  \sin(\omega^{(\lambda)} + \omega), \quad
\omega^{(\lambda)}_\eta = -\omega_\eta + \frac1\lambda
  \sin(\omega^{(\lambda)} - \omega).
$$
Here $\lambda$ is a constant called the B\"acklund parameter.

The B\"acklund transformation is particularly useful in combination with Bianchi's
permutability theorem~\cite{Bia b}; 
see also, e.g.,~\cite{Bia II,R-S}.
To simplify exposition, we shall write $\mathcal B^{(\lambda)}_c \omega$
to denote a solution $\omega^{(\lambda)}$ of system~\eqref{BT} for a specified 
value of the integration constant~$c$.
The Bianchi permutability theorem says that given a pair of B\"acklund
parameters $\lambda_1 \ne \lambda_2$, then for every choice of integration 
constants $c_1,c_2$ there is a unique choice of integration constants 
$c_1', c_2'$ such that
$$
\numbered\label{PT}
\mathcal B^{(\lambda_2)}_{c_2'} \mathcal B^{(\lambda_1)}_{c_1} \omega
 = \mathcal B^{(\lambda_1)}_{c_1'} \mathcal B^{(\lambda_2)}_{c_2} \omega
$$
and, moreover, denoting by $\omega^{(\lambda_1\lambda_2)}$ the common value 
in~\eqref{PT}, then
$\omega^{(\lambda_1\lambda_2)}$ can be obtained from the {\it superposition principle}
$$
\numbered\label{SP}
\tan\frac{\omega^{(\lambda_1\lambda_2)} - \omega}2
 = \frac{\lambda_1 + \lambda_2}{\lambda_1 - \lambda_2}
 \tan\frac{\omega^{(\lambda_1)} - \omega^{(\lambda_2)}}2
$$
without further integration, by purely algebraic manipulations.
The corresponding pseudospherical surfaces 
$\mathbf r^{(\lambda_1 \lambda_2)}$ 
can be easily obtained by iterating formula~\eqref{BTr}.

Assume now that a general solution of system~\eqref{BT} is known for every value 
of the B\"acklund parameter $\lambda$.
Substituting $\omega = \omega^{(\lambda_1)}$ into formula~\eqref{SP}, one can also 
compute the B\"acklund transform
$\omega^{(\lambda_1\lambda_2\lambda_3)}
 = \mathcal B^{(\lambda_3)} \omega^{(\lambda_1\lambda_2)}
 = \mathcal B^{(\lambda_2)} \omega^{(\lambda_1\lambda_3)}$,
by purely algebraic manipulations. 
In principle, this process may be repeated indefinitely, leading to solutions 
$\omega^{(\lambda_1\lambda_2 \dots \lambda_s)}$
depending on any 
finite number of B\"acklund parameters and integration constants, which are thereby
obtained by purely algebraic manipulations.
Needless to say, the corresponding pseudospherical surfaces 
$\mathbf r^{(\lambda_1 \lambda_2 \dots \lambda_s)}$ 
can be computed by iterating the formula~\eqref{BTr}.
Having summarized the B\"acklund transformation and the Bianchi superposition 
principle, we proceed to generation of surfaces of constant astigmatism by 
purely algebraic means.

To start with, we remind the reader that in the particular case of $\lambda = 1$ the 
B\"acklund transformation coincides with Bianchi's~\cite{Bia a} complementarity 
relation, cf.~\eqref{comp evolutes} (actually,
the B\"acklund transformation is a combination of complementarity and Lie's
transformation, and the latter is identity if $\lambda = 1$).
Otherwise said, the complementary pseudospherical surfaces result from the particular 
case $\mathcal B^{(1)}$ of the B\"acklund transformation.
Consequently, the superposition formula~\eqref{SP} yields a method to obtain abundant
pairs of complementary sine-Gordon solutions 
$\omega^{(\lambda_1\lambda_2 \dots \lambda_s)}$
and $\omega^{(\lambda_1\lambda_2 \dots \lambda_s 1)}$.
Likewise, one can also obtain abundant pairs of complementary 
pseudospherical surfaces $\mathbf r^{(\lambda_1 \lambda_2 \dots \lambda_s)}$
and $\mathbf r^{(\lambda_1 \lambda_2 \dots \lambda_s 1)}$ by using formula~\eqref{BTr}.

Let $\mathbf r$ be a pseudospherical surface, corresponding to a sine-Gordon 
solution~$\omega$.
Substituting $\lambda = 1$ into formulas~\eqref{BTr} and~\eqref{BT}, we immediately 
see that the complementary surface is
$$
\numbered\label{compl}
\mathbf r^{(1)} = \mathbf r
 + \frac{\sin(\omega - \omega^{(1)})}{\sin(2\omega)} \mathbf r_\xi
 + \frac{\sin(\omega + \omega^{(1)})}{\sin(2\omega)} \mathbf r_\eta,
$$
where $\omega^{(1)}$ is the {\it complementary solution} of the sine-Gordon equation,
satisfying the compatible first-order equations
$$
\numbered\label{BT1}
\omega^{(1)}_\xi = \omega_\xi + \sin(\omega^{(1)} + \omega), \quad
\omega^{(1)}_\eta = -\omega_\eta + \sin(\omega^{(1)} - \omega).
$$
Before proceeding further, we recall two important observations due to Bianchi~\cite{Bia a} 
(see also~\cite[\S386]{Bia II}).
Considering the dependence of $\omega^{(1)}$ on the integration constant $c$ and
denoting $f = \ln(`d\omega^{(1)}/`d c)$, differentiation of~\eqref{BT1} gives
$$
\numbered\label{BTf}
f_\xi = \cos(\omega^{(1)} + \omega), \quad
f_\eta = \cos(\omega^{(1)} - \omega).
$$
Similarly, taking one more derivative $f' = `d f/`d c$, we get
$$
\numbered\label{BTf'}
f'_\xi = -`e^f \sin(\omega^{(1)} + \omega), \quad  
f'_\eta = -`e^f \sin(\omega^{(1)} - \omega).
$$
It follows that, knowing solutions of system~\eqref{BT1}, we can also
obtain solutions of systems~\eqref{BTf} and~\eqref{BTf'} by purely algebraic 
manipulations and differentiation.

All this is important because surfaces of constant astigmatism are easier to 
obtain from a pair of complementary pseudospherical surfaces $\mathbf r$ and 
$\mathbf r^{(1)}$ than from a single pseudospherical surface (as considered
in~\cite{B-M I}).
Denote
$$
\numbered\label{nf}
\tilde{\mathbf n} = \mathbf r^{(1)} - \mathbf r
 = \frac{\sin(\omega - \tilde\omega)}{\sin(2\omega)} \mathbf r_\xi
 + \frac{\sin(\omega + \tilde\omega)}{\sin(2\omega)} \mathbf r_\eta.
$$
Then $\tilde{\mathbf n}$ 
is a unit vector tangent to both surfaces
$\mathbf r$ and $\mathbf r^{(1)}$ and determines what is called a pseudospherical
congruence. Normal surfaces of this congruence are the constant astigmatism 
surfaces sought.

\begin{proposition} \label{prop:1}
Let $\omega^{(1)}(\xi,\eta,c)$ be a general solution of system~\eqref{BT1}, where~$c$ 
is an integration constant. 
Then 
$\tilde{\mathbf r} = \mathbf r - f \tilde{\mathbf n}$, 
where $f = \ln(`d\omega^{(1)}/`d c)$ and
$\tilde{\mathbf n}$ is the unit vector given by
formula~\eqref{nf}, is a surface of constant 
astigmatism having surfaces $\mathbf r$ and $\mathbf r^{(1)}$ as evolutes.
\end{proposition}

\begin{proof}
The surface $\tilde{\mathbf r} = \mathbf r - f \tilde{\mathbf n}$ is normal 
to the congruence determined by the surface $\mathbf r$ and vectors $\tilde{\mathbf n}$, 
if and only if 
$\tilde{\mathbf r}_\xi \cdot \tilde{\mathbf n}
 = \tilde{\mathbf r}_\eta \cdot \tilde{\mathbf n} = 0$.
By virtue of the Gauss--Weingarten system~\eqref{GW:SG} above, the derivatives of
$\tilde{\mathbf r} = \mathbf r - f \tilde{\mathbf n}$ can be written as 
$$
\tilde{\mathbf r}_\xi  
 = (1
    + \frac{\sin 2\omega^{(1)} + \sin 2\omega}{2 \sin 2\omega} f
    + \frac{\sin(\omega^{(1)} - \omega)}{\sin 2\omega} f_\xi) \mathbf r_\xi
\\\qquad
 - (\frac{\sin 2(\omega^{(1)} + \omega)}{2 \sin 2\omega} f
    + \frac{\sin(\omega^{(1)} + \omega)}{\sin 2\omega} f_\xi) \mathbf r_\eta
 - \sin(\omega^{(1)} + \omega) f \mathbf n,
\\
\tilde{\mathbf r}_\eta  
 = (1
    - \frac{\sin 2\omega^{(1)} - \sin 2\omega}{2 \sin 2\omega} f
    - \frac{\sin(\omega^{(1)} + \omega)}{\sin 2\omega} f_\eta) \mathbf r_\eta
\\\qquad
 + (\frac{\sin 2(\omega^{(1)} - \omega)}{2 \sin 2\omega} f
    - \frac{\sin(\omega^{(1)} - \omega)}{\sin 2\omega} f_\eta) \mathbf r_\xi
 + \sin(\omega^{(1)} - \omega) f \mathbf n,
$$
Now it is straightforward to check that the conditions
$\tilde{\mathbf r}_\xi \cdot \tilde{\mathbf n} = 0$ and 
$\tilde{\mathbf r}_\eta \cdot \tilde{\mathbf n} = 0$ 
reduce to equations~\eqref{BTf}.
Moreover, it is a routine to verify that $\tilde{\mathbf r}$ is of constant 
astigmatism equal to~$1$.
Actually an equivalent computation will be done in the proof of the next proposition
under the same assumptions.
\end{proof}

Based on Bianchi's observation above, Proposition~\ref{prop:1} shows that the 
constant astigmatism surfaces $\tilde{\mathbf r} = \mathbf r - f \tilde{\mathbf n}$ 
can be found by purely algebraic manipulations and differentiation
once a one-parameter family of pseudopotentials $\omega^{(1)}$ is known.

However, since $\xi,\eta$ need not be curvature coordinates on the constant 
astigmatism surfaces $\tilde{\mathbf r} = \mathbf r - f \tilde{\mathbf n}$, 
Proposition~\ref{prop:1}, as it stands, yields neither a solution of the constant 
astigmatism equation nor an orthogonal equiareal pattern on the 
sphere~$\tilde{\mathbf n}$. Yet the coordinates $\xi,\eta$ have a geometric meaning
of a slip line field according to Definition~\ref{def:slip line}.

\begin{proposition}
\label{prop:y}
Let $\omega^{(1)}(\xi,\eta,c)$ be a general solution of system~\eqref{BT1}, where~$c$ 
is an integration constant, let $f = \ln(`d\omega^{(1)}/`d c)$ and
$x = \dif f/\dif c$. Let 
$y(\xi,\eta)$ be a solution of the system 
$$
\numbered\label{BTh}
y_\xi = `e^{-f} \sin(\omega + \omega^{(1)}), \quad
y_\eta = `e^{-f} \sin(\omega - \omega^{(1)}).
$$
Then $x,y$ are adapted curvature coordinates on the surface 
$\tilde{\mathbf r}$.
Moreover, if $z = `e^{-2 f}$, then $z(x,y)$ is a solution of the constant astigmatism 
equation~\eqref{CAE}.
Finally, $z \dif x^2 + \dif y^2/z$ is an orthogonal equiareal pattern on the unit sphere 
$\tilde{\mathbf n}$, while $\xi,\eta$ is the associated slip line field
according to Definition~\ref{def:slip line}.
\end{proposition}

\begin{proof}
Continuing the routine computations started in the proof of Proposition~\ref{prop:1} we 
obtain
$$
\tilde{\mathbf I} = \frac12 (1 + 2 f + 2 f^2)
  (1 - \cos 2(\omega + \omega^{(1)})) \dif\xi^2
\\\quad  + (1 + 2 f) (\cos 2\omega - \cos 2\omega^{(1)})\dif\xi\dif\eta
\\\quad  + \frac12 (1 + 2 f + 2 f^2)
  (1 - \cos 2(\omega - \omega^{(1)})) \dif\xi^2,
\\
\tilde{\mathbf{II}} = \frac12 (1 + 2 f)
  (1 - \cos 2(\omega + \omega^{(1)})) \dif\xi^2
\\\quad  + (\cos 2\omega - \cos 2\omega^{(1)})\dif\xi\dif\eta
\\\quad  + \frac12 (1 + 2 f)
  (1 - \cos 2(\omega - \omega^{(1)})) \dif\xi^2.
$$
The corresponding shape operator is
$$
\frac1{2 f (f + 1)}
(\begin{array}{cc}
2 f + 1 & \frac{1 - \cos 2(\omega^{(1)} - \omega)}{\cos 2\omega^{(1)} - \cos 2\omega} \\
\frac{1 - \cos 2(\omega^{(1)} + \omega)}{\cos 2\omega^{(1)} - \cos 2\omega} & 2 f + 1
\end{array}).
$$
Its eigenvalues are the curvatures, namely $1/f$ and $1/(f + 1)$.
We choose $\rho_1 = f + 1$, $\rho_2 = f$ to have $\rho_1 - \rho_2 = 1$.
The eigenvectors yield two principal directions 
$$
\Pi_\pm = (\cos 2\omega^{(1)} - \cos 2\omega) \frac\partial{\partial\xi}
  \pm (1 - \cos 2(\omega + \omega^{(1)})) \frac\partial{\partial\eta}.
$$
It is easy to check that $x$ and $y$ are integrals of the fields $\Pi_+$ and
$\Pi_-$, respectively.
This implies that $x,y$ are curvature coordinates.
Now, equations
$$
\tilde{\mathbf{I}} = u^2 \dif x^2 + v^2 \dif y^2, \quad 
\tilde{\mathbf{II}} = (u^2/\rho_1) \dif x^2 + (v^2/\rho_2) \dif y^2.
$$
yield 
$$
\numbered\label{uv}
u = (f + 1)/`e^f, \quad  v = f `e^f.
$$
Substituting into condition~\eqref{adapted}, we see that
the curvature coordinates $x,y$ are adapted.
Moreover, it is straightforward to verify that $z = `e^{-2 f}$ satisfies the 
equation of constant astigmatism~\eqref{CAE} with respect to independent 
variables~$x,y$.

To prove that $\xi,\eta$ is the associated slip line field on the sphere 
$\tilde{\mathbf n}$, it suffices to check that 
$$
\tilde{\mathbf n}_\xi =
 \frac{\cos 3\omega - \cos\omega - \cos(2\omega^{(1)} - \omega)
  + \cos(2\omega^{(1)} + \omega)}{8 \sin^2 \omega \cos\omega} \mathbf r_\xi
\\\qquad
 + \frac{\cos(2\omega^{(1)} + \omega) - \cos(2\omega^{(1)} + 3\omega)}
    {8 \sin^2 \omega \cos\omega} \mathbf r_\eta
 + \sin(\omega^{(1)} + \omega)\,\mathbf n
\\
\tilde{\mathbf n}_\eta =
 \frac{\cos 3\omega - \cos\omega + \cos(2\omega^{(1)} - \omega)
  - \cos(2\omega^{(1)} + \omega)}{8 \sin^2 \omega \cos\omega} \mathbf r_\eta
\\\qquad
 + \frac{\cos(2\omega^{(1)} - \omega) - \cos(2\omega^{(1)} - 3\omega)}
    {8 \sin^2 \omega \cos\omega} \mathbf r_\xi
 - \sin(\omega^{(1)} - \omega)\,\mathbf n
$$
bisect the right angle between
$$
\tilde{\mathbf n}_x =
 -\frac{\cos(\omega^{(1)} - \omega)}{\sin 2\omega `e^f} \mathbf r_\xi
 + \frac{\cos(\omega^{(1)} + \omega)}{\sin 2\omega `e^f} \mathbf r_\eta
\quad\text{and}\quad
\tilde{\mathbf n}_y = `e^f \mathbf n,
$$
according to Definition~\ref{def:slip line}. This is straightforward.
\end{proof}

\begin{corollary}
If $S$ is a constant astigmatism surface, then the asymptotic coordinates 
on the focal surfaces of $S$ correspond to slip line fields on the Gaussian 
image of $S$.
\end{corollary}

Proposition~\ref{prop:y} allows us to construct one of the adapted curvature coordinates 
by purely algebraic manipulations and differentiation, while
the other curvature coordinate has to be obtained by integration. 
It is therefore natural to ask whether one could obtain superposition
formulas for $f,x,y$ similar to formula~\eqref{SP}.
The answer is positive.


\begin{definition} \rm
Given two sine-Gordon solutions $\omega$ and $\omega^{(\lambda)}$ related by the 
B\"acklund transformation $\mathcal B^{(\lambda)}$, let 
$f^{(\lambda)}$,
$x^{(\lambda)}$, 
$y^{(\lambda)}$ 
denote functions satisfying the compatible equations
$$
\numbered\label{fxy}
\begin{gathered}
f^{(\lambda)}_\xi = \lambda \cos(\omega^{(\lambda)} + \omega),
& 
f^{(\lambda)}_\eta = \frac1\lambda \cos(\omega^{(\lambda)} - \omega),
\\
x^{(\lambda)}_\xi = \lambda `e^{f^{(\lambda)}} \sin(\omega^{(\lambda)} + \omega),
& 
x^{(\lambda)}_\eta = \frac1\lambda `e^{f^{(\lambda)}} \sin(\omega^{(\lambda)} - \omega),
\\
y^{(\lambda)}_\xi = \lambda `e^{-f^{(\lambda)}} \sin(\omega^{(\lambda)} + \omega),
& 
y^{(\lambda)}_\eta = -\frac1\lambda `e^{-f^{(\lambda)}} \sin(\omega^{(\lambda)} - \omega).
\end{gathered}
$$
The quantities $f^{(\lambda)}$, $x^{(\lambda)}$, $y^{(\lambda)}$ 
will be called {\it associated potentials} corresponding to the pair
$\omega,\omega^{(\lambda)}$.
\end{definition}

\begin{proposition}
Let $\omega, \omega^{(\lambda_1)}, \omega^{(\lambda_2)}, 
 \omega^{(\lambda_1\lambda_2)}$ be four sine-Gordon solutions  
related by the Bianchi superposition principle~\eqref{SP}.
Then the associated potentials 
$f^{(\lambda_1\lambda_2)}$, $x^{(\lambda_1\lambda_2)}$, $y^{(\lambda_1\lambda_2)}$
corresponding to the pair $\omega^{(\lambda_1)},\omega^{(\lambda_1\lambda_2)}$
are related to the associated potentials 
$f^{(\lambda_2)}$, $x^{(\lambda_2)}$, $y^{(\lambda_2)}$ corresponding to the pair 
$\omega,\omega^{(\lambda_2)}$ by formulas
$$
\numbered\label{SPp}
f^{(\lambda_1\lambda_2)} = f^{(\lambda_2)}
 - \ln(2 \cos(\omega^{(\lambda_1)} - \omega^{(\lambda_2)})
       - \frac{\lambda_1}{\lambda_2} - \frac{\lambda_2}{\lambda_1}),
\\
x^{(\lambda_1\lambda_2)}
 = \frac{\lambda_1 \lambda_2}{\lambda_1^2 - \lambda_2^2} 
   (x^{(\lambda_2)}
     -  
      \frac{2 \lambda_1 \lambda_2 \sin(\omega^{(\lambda_1)} - \omega^{(\lambda_2)})}
         {\lambda_1^2
          - 2 \lambda_1 \lambda_2 \cos(\omega^{(\lambda_1)} - \omega^{(\lambda_2)})
          + \lambda_2^2} `e^{f^{(\lambda_2)}}),
\\
y^{(\lambda_1\lambda_2)}
 = (\frac{\lambda_1}{\lambda_2} - \frac{\lambda_2}{\lambda_1}) y^{(\lambda_2)}
   - 2 `e^{-f^{(\lambda_2)}} \sin(\omega^{(\lambda_1)} - \omega^{(\lambda_2)}),
$$
up to an additive constant.
\end{proposition}

\begin{proof}
It is straightforward to check that 
$f^{(\lambda_1\lambda_2)}$, $x^{(\lambda_1\lambda_2)}$, $y^{(\lambda_1\lambda_2)}$
given by formulas~\eqref{SPp} satisfy
$$ 
\hskip-3pc
\begin{gathered}
f^{(\lambda_1\lambda_2)}_\xi
 = \lambda_2 \cos(\omega^{(\lambda_1\lambda_2)} + \omega^{(\lambda_1)}),
& 
f^{(\lambda_1\lambda_2)}_\eta
 = \frac1{\lambda_2} \cos(\omega^{(\lambda_1\lambda_2)} - \omega^{(\lambda_1)}),
\\
x^{(\lambda_1\lambda_2)}_\xi
 = \lambda_2 `e^{f^{(\lambda)}}
   \sin(\omega^{(\lambda_1\lambda_2)} + \omega^{(\lambda_1)}),
& %
x^{(\lambda_1\lambda_2)}_\eta
 = \frac1{\lambda_2} `e^{f^{(\lambda_1\lambda_2)}}
   \sin(\omega^{(\lambda_1\lambda_2)} - \omega^{(\lambda_1)}),
\\
y^{(\lambda_1\lambda_2)}_\xi
 = \lambda_2 `e^{-f^{(\lambda_1\lambda_2)}}
   \sin(\omega^{(\lambda_1\lambda_2)} + \omega^{(\lambda_1)}),
& 
y^{(\lambda_1\lambda_2)}_\eta
 = -\frac1{\lambda_2} `e^{-f^{(\lambda_1\lambda_2)}}
   \sin(\omega^{(\lambda_1\lambda_2)} - \omega^{(\lambda_1)}).
\end{gathered}
$$
whenever
$f^{(\lambda_2)}$, $x^{(\lambda_2)}$, $y^{(\lambda_2)}$ satisfy~\eqref{fxy} with
$\lambda = \lambda_2$.
\end{proof}

\begin{example} \rm
({\it One-soliton solutions})
Let us apply the procedure outlined above to the one-soliton solutions 
$$
\omega^{(\lambda)} = \mathcal B^{(\lambda)}_{0} (0) = 2 \arctan \exp p_\lambda, 
$$
of the sine-Gordon equation.
Here and in what follows we denote
$$
p_\lambda = \lambda \xi + \frac \eta\lambda,
\quad
q = \xi - \eta.
$$
As is well known, these one-soliton solutions correspond to the Dini surfaces
(helicoids of the tractrix)
$$
\mathbf r^{(\lambda)} = \frac{2 \lambda}{1 + \lambda^2} 
 (\begin{array}{c}
   `sech p_\lambda \sin q \\
   `sech p_\lambda \cos q \\
   p_\lambda - `tanh p_\lambda
  \end{array})
 + \frac{1 - \lambda^2}{1 + \lambda^2} 
   (\begin{array}{c} 0 \\ 0 \\ q\end{array}).
$$
We now proceed to the complementary surfaces of the Dini surfaces, which
correspond to the nonlinear superposition of $\omega^{(\lambda)}$ and $\omega^{(1)}$,
i.e., the two-soliton solutions
$$
\omega^{(\lambda 1)} = \mathcal B^{(1)}_{c} (\omega^{(\lambda)})
 = 2 \arctan
      \frac{(\lambda + 1)(`e^{p_\lambda} - `e^{p_1 + c})}
           {(\lambda - 1)(1 + `e^{p_\lambda + p_1 + c})},
$$ 
where, obviously, $p_1 = \xi + \eta$.
The particular case of $\lambda = 1$ (the Beltrami pseudosphere) is excluded from
consideration. 


After tedious computations, one obtains the resulting quantities 
$x = x^{(\lambda 1)}$, $y = y^{(\lambda 1)}$, 
$z = `e^{-2 f^{\lambda 1}}$. They are
$$
\hskip-3pc
\begin{gathered}
x = \frac{\lambda}{\lambda^2 - 1} 
\times\frac
{(\lambda - 1)^2 (c_2 A^2 B^2 - c_1)
 - (\lambda + 1)^2 (c_1 B^2 + c_2 A^2)
 + 4 (c_1 - c_2) \lambda A B
}
{(\lambda - 1)^2 (A^2 B^2 + 1)
 + (\lambda + 1)^2 (B^2 + A^2)
 - 8 \lambda A B
}, 
\\
y = \frac{4 \ln B}{c_1 + c_2}
 - 2 \frac{(\lambda^2 + 1) \ln A}{(c_1 + c_2) \lambda}
 + \frac
{4 \lambda (A B + 1)(A - B)
 + c_3 (c_1 + c_2) (\lambda^2 - 1) A (1 + B^2) 
}
{(c_1 + c_2) \lambda `e^s (1 + `e^{2p_\lambda})},
\\
z = 
(\frac
{(\lambda - 1)^2 (A^2 B^2 + 1)
 + (\lambda + 1)^2 (B^2 + A^2)
 - 8 \lambda A B
}{(c_1 + c_2) \lambda A (1 + B^2)})^2,
\end{gathered}
$$ 
where $A = `e^{p_1} = `e^{\xi + \eta}$ and 
$B = `e^{p_\lambda} = `e^{\lambda \xi + \eta/\lambda}$, 
while $c_1,c_2,c_3$ are arbitrary constants.
By eliminating $\xi,\eta$ one obtains
$$
\numbered\label{y(x,z)}
\hskip-1pc
y = \frac{1}{c_1 + c_2}
 \left(\frac{4 (AB+1) (A-B)}{(B^2+1) A}
  - 2 \frac{\lambda^2+1}{\lambda} \ln A
  + 4 \ln B
\right)
 + \frac{(\lambda^2-1) c_3}{\lambda},
$$
where 
$$
A = \frac{\lambda (\lambda^2+1) (c_1+c_2) \sqrt{z} - \sqrt{k}}
{(\lambda^2-1)^2 + (\lambda^2 x - \lambda c_2 - x)^2 z},
\\
B = \frac{2 \lambda^2 (c_1+c_2)  \sqrt{z} + \sqrt{k}}
{(\lambda^2-1)^2 + (\lambda^2 x - \lambda c_2 - x) (\lambda^2 x+\lambda c_1-x) z},
\\
k = -\left[(\lambda^2-1)^2 + 2 (c_1+c_2) \lambda^2 \sqrt{z} + (\lambda^2 x-\lambda c_2-x) (\lambda^2 x+\lambda c_1-x) z\right]
\\
\qquad \times \left[(\lambda^2-1)^2 - 2 (c_1+c_2) \lambda^2 \sqrt{z} + (\lambda^2 x - \lambda c_2 - x) (\lambda^2 x+\lambda c_1-x) z \right].
$$
Eq.~\eqref{y(x,z)} is an implicit formula for a solution $z(x,y)$ of the constant astigmatism equation.
Using Proposition \ref{prop:1} it is now easy to 
construct the surface of constant astigmatism from its two evolutes 
$\mathbf{r}^{(\lambda)}$ and $\mathbf{r}^{(\lambda 1)}$ 
as well as the orthogonal slip line net on the Gaussian sphere $\tilde{\mathbf{n}} = \mathbf{r}^{(\lambda 1)} - \mathbf{r}^{(\lambda)}$, part of 
which can be seen on Fig.~\ref{fig:features} 
(actually, the sphere is multiply covered).
This example also demonstrates that sphere's slip line fields are prone to 
developing singularities.

\begin{figure}
\centering
\includegraphics[width = 7.5cm]{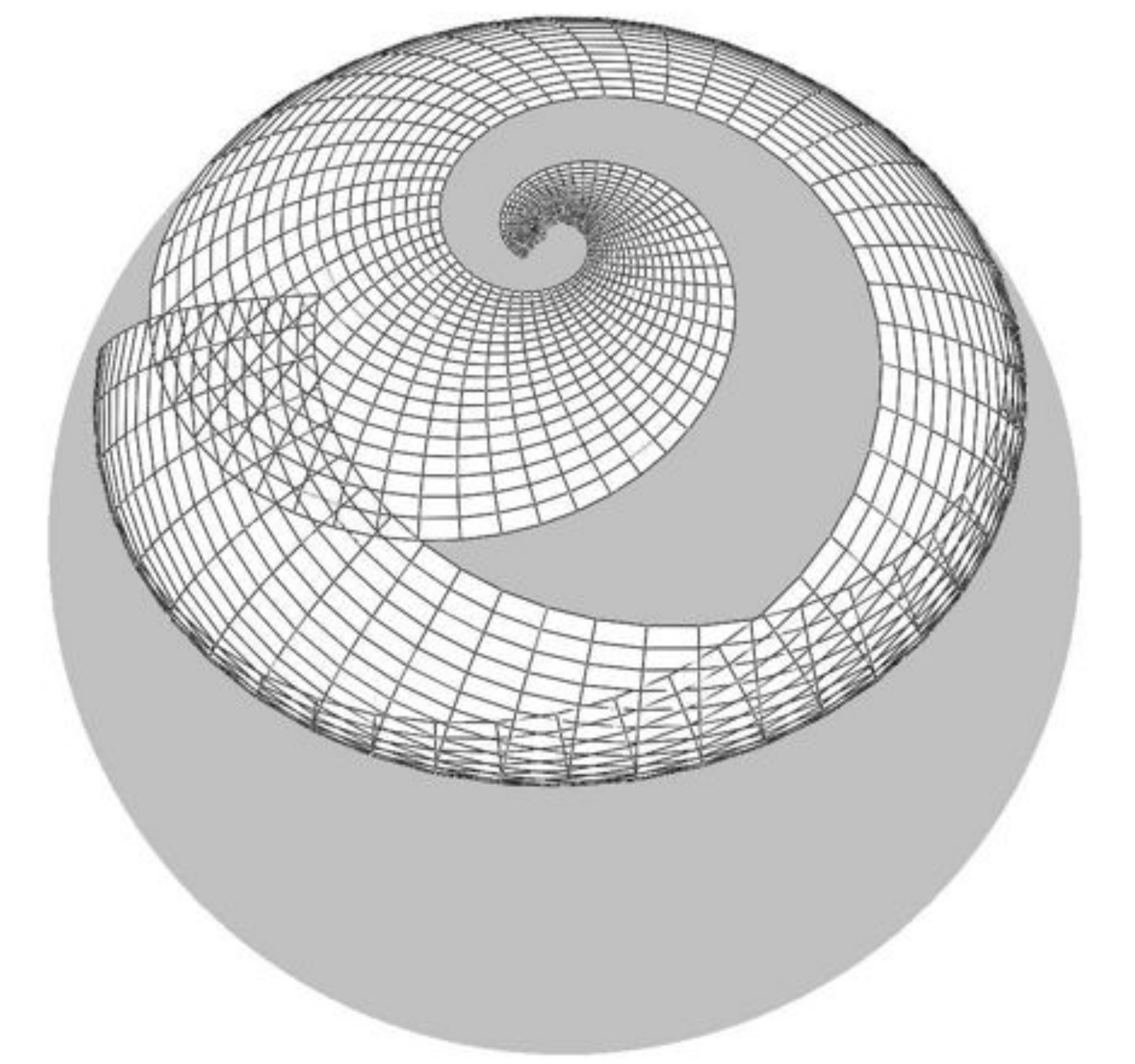}
\caption{Sphere's slip line field with features}
\label{fig:features}
\end{figure}
\end{example}

Construction of constant astigmatism surfaces related to the 
sine-Gordon $n$-soliton solutions is postponed to a separate paper.

\section{Lipschitz surfaces in principal coordinates} \rm \label{sect:Lipsch}


In 1887 Lipschitz~\cite{Lip} presented a class of surfaces of constant
astigmatism in terms of spherical coordinates related to the Gaussian image. 
To find the corresponding solutions of the constant astigmatism equation, 
one has to obtain the surfaces in terms of the pricipal coordinates.
Redoing the computation is easier than transforming the Lipschitz result.

Consider the unit sphere 
$\mathbf{n} = (\cos\lon \sin\lat, \sin\lon \sin\lat, \cos\lat)$ 
parameterised by the latitude $\lat$ and longitude $\lon$. 
To specify an orthogonal equiareal pattern we let $\lat,\lon$ denote yet unknown 
functions of parameters $x,y$. Lipschitz defines a {\it Stellungswinkel\/} to 
be the angle $\omega$ between 
$\mathbf{n}_\lat$ and 
$\mathbf{n}_x = \lon_x \mathbf{n}_\lon + \lat_x \mathbf{n}_\lat$.
The Lipschitz class is specified by allowing the Stellungswinkel to 
depend solely on the latitude~$\theta$.

\begin{theorem}
The general Lipschitz solution of the constant astigmatism equation~\eqref{CAE} 
depends on four constants $h_{11},h_{10},h_{01},h_{00}$
and consists of functions
$$
\numbered\label{Lipsch z}
z = \frac{1 - h^2 + \sqrt{(1 - h^2)^2 - 4 (H_1 h - H_2)^2}}{2 (h_{11} x + h_{01})^2},
$$ 
where
$$
h = h_{11} x y  + h_{10} x + h_{01} y + h_{00}, \quad 
H_1 = h_{11}, \quad H_2 = h_{11} h_{00} - h_{10} h_{01}.
$$
The constants $H_1,H_2$ are invariant with respect to 
the translations $\xTranslation,\yTranslation$, the scaling $\Scaling$
and the involution $\Involution$. 
Formula~\eqref{Lipsch z} covers all Lipshitz solutions except a particular 
solution
$$
z = \frac1{c_1 - (x - c_0)^2},
$$ 
$c_1,c_0$ being arbitrary constants.
\end{theorem}

\begin{proof}
Computing the first fundamental form in two ways
we have
$$
\dif\lat^2 + \sin^2\lat \dif\lon^2 = z \dif x^2 + (1/z) \dif y^2,
$$
i.e.,
$$
\lat_x^2 + \lon_x^2 \sin^2 \lat = z, \quad
\lat_x \lat_y + \lon_x \lon_y \sin^2 \lat = 0, \quad
\lat_y^2 + \lon_y^2 \sin^2 \lat = 1/z.
$$
Eliminating $z$, we have 
$$
\numbered\label{Lipsch1}
(\lat_x^2 + \lon_x^2 \sin^2 \lat)(\lat_y^2 + \lon_y^2 \sin^2 \lat) = 1, \quad 
\lat_x \lat_y + \lon_x \lon_y \sin^2 \lat = 0.
$$
Since
$$
\mathbf{n}_\lat = (\cos\lon \cos\lat, \sin\lon \cos\lat, -\sin\lat), \quad
\mathbf{n}_\lon = (-\sin\lon \sin\lat, \cos\lon \sin\lat, 0), \\
\mathbf{n}_x = (\lat_x \cos\lon \cos\lat - \lon_x \sin\lon \sin\lat,
  \lat_x \sin\lon \cos\lat + \lon_x \cos\lon \sin\lat, -\lat_x \sin\lat),
$$
we easily compute 
$$
\numbered\label{Lipsch co}
\cos\omega = \frac{\mathbf{n}_\lat \cdot \mathbf{n}_x}{|\mathbf{n}_\lat|\,|\mathbf{n}_x|}
 = \frac{\lat_x}{\sqrt{\lat_x^2 + \lon_x^2 \sin^2 \lat}}.
$$
Honce, the Lipschitz' condition amounts to $\lon_x/\lat_x$ being a 
function of~$\lat$ alone. This can be conveniently written in the form
$$
\numbered\label{Lipsch2}
\lon_x = \frac{\Theta(\lat)}{\sin\lat} \lat_x
$$ 
(intentionally leaving $\sin\lat$ unabsorbed).
Conditions~\eqref{Lipsch1} and~\eqref{Lipsch2} combine into
the system 
$$
\numbered\label{Lipsch3}
\lon_x = \frac{\Theta}{\sin\lat} \lat_x, \quad
\lon_y = -\frac{1}{(\sin\lat)(\Theta^2 + 1) \lat_x}, \quad
\lat_y = \frac{\Theta}{(\Theta^2 + 1) \lat_x}.
$$
Computing the compatibility conditions for $\lon$ we get
$(\sin\lat) \lat_{xx} + (\cos\lat) \lat_x^2 = 0$
with the general solution
$$
\lat = \arccos h, \quad h = h_1 x + h_0,
$$
$h_1,h_0$ being arbitrary functions of $y$.
Now the third equation of~\eqref{Lipsch3} implies
$$
\numbered\label{Lipsch Th}
\frac{\Theta}{\Theta^2 + 1} = \lat_x \lat_y
 = \frac{h_1}{1 + (h_1 x + h_0)^2}
   (\frac{\partial h_1}{\partial y} x + \frac{\partial h_0}{\partial y}).
$$
Since the left-hand side is a function of $\lat$ alone,
the same is true for the right-hand side.
Checking the Jacobian, we get
$$
\numbered\label{Lipsch Jac}
0 = \frac{h_1^2}{(1 + (h_1 x + h_0)^2)^{3/2}}
 (\frac{\partial h_1^2}{\partial y^2} x + \frac{\partial h_0^2}{\partial y^2}).
$$
If $h_1 = 0$, then $h = h_0$ is a function of $y$ alone and from 
system~\eqref{Lipsch3} we easily obtain $h_0 = h_{01} y + h_{00}$ and
$z = 1/h_{01}^2 - (y + h_{00}/h_{01})^2$, which is the particular solution.

Otherwise $h_1 \ne 0$; it follows from \eqref{Lipsch Jac} that $h_1,h_0$ are 
linear in $y$ and we can write
$$
\numbered\label{Lipsch h}
h = h_{11} x y  + h_{10} x + h_{01} y + h_{00},
$$ 
with $h_{ij}$ being arbitrary constants.
Now, eq.~\eqref{Lipsch Th} gives
$$
\frac{\Theta}{\Theta^2 + 1}
 = \frac{H_1 h - H_2}{1 - h^2},
\quad\text{where}\quad
H_1 = h_{11}, \quad H_2 = h_{11} h_{00} - h_{10} h_{01}.
$$
Hence,
$$
\Theta = \frac{1 - h^2 + \sqrt{(1 - h^2)^2 - 4 (H_1 h - H_2)^2}}{2 (H_1 h - H_2)}.
$$
Since $H_1,H_2$ are constants, $\Theta$ is a function of~$h$ and, consequently, of~$\lat$
alone.
Inserting into system~\eqref{Lipsch3} and then into $z = \lat_x^2 + \lon_x^2 \sin^2 \lat$, 
we obtain the general solution~\eqref{Lipsch z}.
\end{proof}

It is now easy to obtain the corresponding orthogonal equiareal pattern.

\begin{theorem}
The orthogonal equiareal pattern corresponding to the general Lipschitz 
solution is $\mathbf{n} = (\cos\lon \sin\lat, \sin\lon \sin\lat, \cos\lat)$, where
$$
\numbered\label{Lipsch oep}
\lat = \arccos h, \\
\lon = - \frac{\ln(h_{11}x + h_{01})}{h_{11}}
 + \int \frac{1 - h^2 + \sqrt{(1 - h^2)^2 - 4 (H_1 h - H_2)^2}}
        {2(H_1 h - H_2)(1 - h^2)} \dif h,  \\
h = h_{11} x y  + h_{10} x + h_{01} y + h_{00}, 
\quad H_1 = h_{11}, \quad H_2 = h_{11} h_{00} - h_{10} h_{01}.
$$ 
\end{theorem}

\begin{proof}
We need to know $\lon$, i.e., we have to integrate the first two 
equations~\eqref{Lipsch3}. 
It is easily observed that
$\lon_x h_y - \lon_y h_x + 1 = 0$. Solving this PDE for $\lon$, we obtain
$$
\lon = - \frac{\ln(h_{11}x + h_{01})}{h_{11}} + \Phi(h),
$$
while for $\Phi(h)$ we get
$$
\frac{\dif\Phi}{\dif h} = \frac{1 - h^2 + \sqrt{(1 - h^2)^2 - 4 (H_1 h - H_2)^2}}{2(H_1 h - H_2)(1 - h^2)}.
$$
\end{proof}

The Stellungswinkel $\omega$ is a function of the lattitude
$\lat$ as required; namely
$$
\cos^2\omega = \frac{1}{\Theta^2 + 1}
 = \frac{1 - h^2 + \sqrt{(1 - h^2)^2 - 4 (H_1 h - H_2)^2}}{2(1 - h^2)}
\\\quad
 = \frac{\sin^2\theta + \sqrt{\sin^4\theta - 4 (H_1 \cos\theta - H_2)^2}}
 {2\sin^2\theta}.
$$

\begin{remark} \rm
It is easy to check that the general Lipschitz solution~\eqref{Lipsch z} 
satisfies 
$$
h_{11} \scaling + h_{01} \xtranslation - h_{10} \ytranslation = 0,
$$
where $\xtranslation = z_x$, 
$\ytranslation = z_y$, 
$\scaling = x z_x - y z_y + 2 z$
are generators (see, e.g.,~\cite{B-V-V}) of the Lie symmetries 
$\xTranslation$, 
$\yTranslation$, 
$\Scaling$, respectively.
This means that~\eqref{Lipsch z} is a symmetry-invariant solution of the 
constant astigmatism equation.
\end{remark}

\begin{example}\rm
When the integral in~\eqref{Lipsch oep} can be 
expressed in terms of elementary functions? Assuming that $h_{11}$ is nonzero, 
$h_{10}$ and $h_{01}$ can be removed by shifts, so we set $h_{10} = h_{01} = 0$. 
Consider the expression under the square root in \eqref{Lipsch oep}. 
Its discriminant with respect to $h$ is proportional to
$$
(1 + H_1^2 + 2 H_2) (1 + H_1^2 - 2 H_2) (H_1 - H_2)^2 (H_1 + H_2)^2
\\\quad
 = h_{11}^4 (1 + h_{11}^2 + 2 h_{11} h_{00}) (1 + h_{11}^2 - 2 h_{11} h_{00})
   (1 - h_{00})^2 (1 + h_{00})^2,
$$
which is zero if and only if 
$$
h_{00} = \pm 1 \quad \text{or}\quad h_{00} = \pm \frac{1 + h_{11}^2}{2h_{11}}.
$$ 
In these cases, $\phi$ can be expressed in terms of elementary functions. For $h_{00} = \pm 1$
we have
$$
\phi = 
\mp\frac{\sqrt{1-C^2}}{2C} \ln\left(  \frac{  4\left( 1-2 C^2 \pm h + \sqrt{1-C^2}\sqrt{(h\mp 1)^2-4(C^2 \mp h)}\right)}{h\mp 1}   \right) 
\\\quad
-\frac{\ln(Cx)}{C}
+ \frac{\ln[h^2 - 1+ (h\mp 1)\sqrt{(h\mp 1)^2-4(C^2 \mp h)}]}{2C}   
\\\quad
\pm \frac{1}{2}\arctan\left( \frac{\sqrt{(h\pm 1)^2-4C^2}}{2C} \right) \,, 
$$ 
where $C = H_1 = h_{11}$ is a constant and $h = Cxy + 1$. 
The 
orthogonal equiareal pattern corresponding to $h_{00} = 1$ and
$h_{11} = 1/4$ can be seen on Fig.~\ref{fig:Lipsch}. 

In the second case, when $h_{00} = \pm (1 + h_{11}^2)/2h_{11}$, we obtain
$$
\phi = -\frac{\ln(Cx)}{C} \mp\frac{\ln\left(h\pm C + \sqrt{h^2 \pm 2Ch -C^2 -2}\right)}{2C} 
\\\quad
- \frac{1}{2C} \ln(	\frac{-C^3-3C  \pm h( 3C^2 + 1) 
 - (1-C^2)\sqrt{h^2 \pm 2Ch -C^2 -2} }
{(C^2 \mp 2Ch + 1)^2})	
\\\quad
\pm \frac{1}{2}\arctan(\frac{2(C h\pm 1)\sqrt{h^2 \pm 2Ch -C^2 -2}}{h^2-C^2h^2-C^2\pm 4Ch-3}).	
$$ 
Here $h,C$ have the same meaning as above.
No figure is provided in this case, since $\phi$ and $\theta$ cannot be simultaneously real.
\end{example}

\begin{figure}
\begin{center}
\includegraphics[scale=0.35]{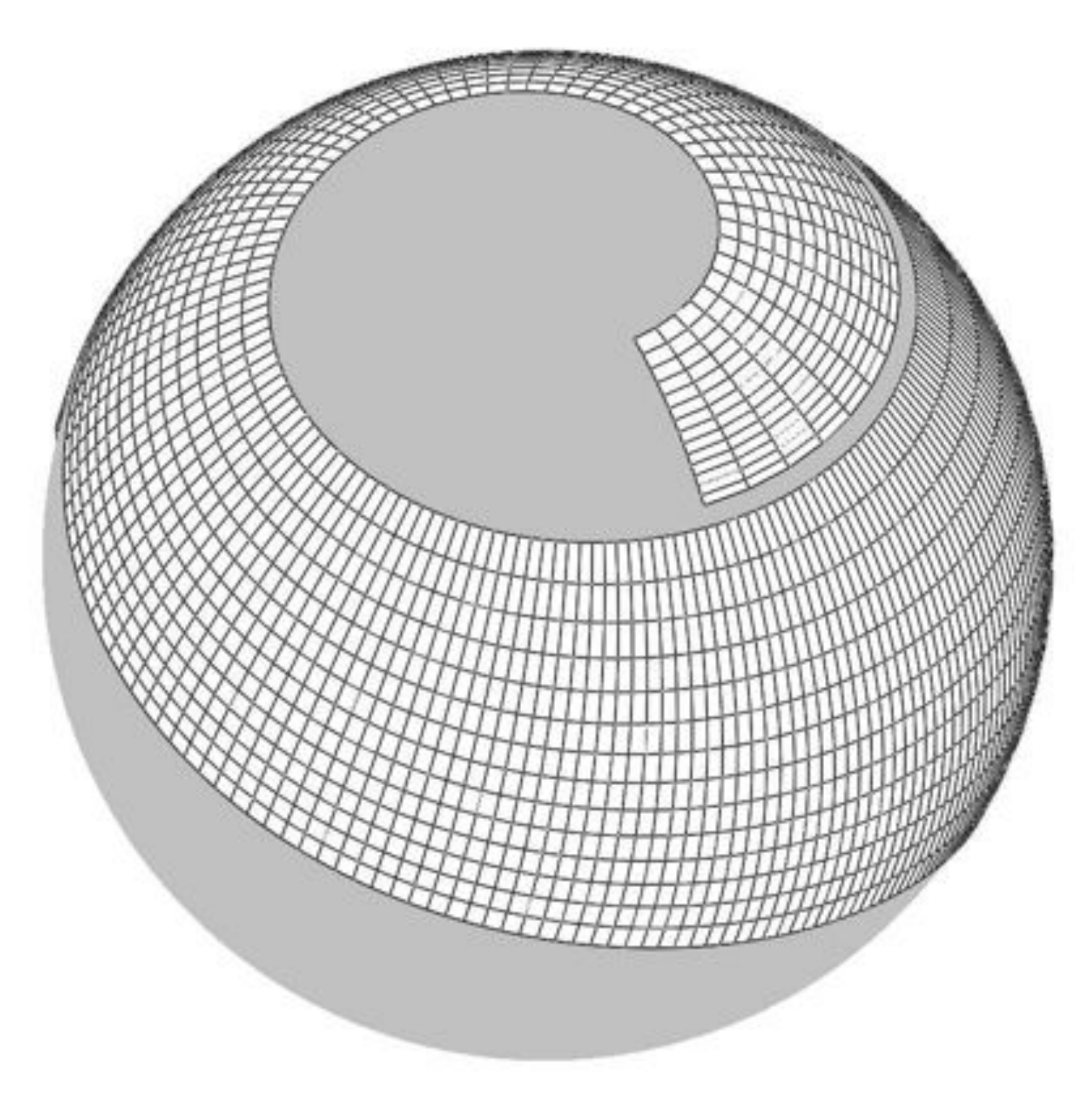}
\caption{The orthogonal equiareal pattern on the sphere corresponding to one 
of the Lipschitz solutions}
\label{fig:Lipsch}
\end{center}
\end{figure}

\section{Conclusions}

Summarizing, we identified the constant astigmatism equation and the sine-Gordon 
equation as integrable models in two-dimensional plasticity on the sphere with 
respect to parameterisation by principal stress lines and slip lines, 
respectively.
We remark in this context that the majority of exact solutions in plasticity 
that can be found in the literature are either due to linearisable systems or 
come from symmetry methods; see, e.g.,~\cite{A-B-S,Lam} and references therein.

We have also extended the classical Bianchi superposition principle so as to
be able to generate solutions of the constant astigmatism equation by 
algebraic manipulations. 
Finally, revisiting the classical Lipschitz surfaces of constant astigmatism,
we have identified them as corresponding to invariant solutions.

In conclusion, we are able to say that obtaining large families of exact 
solutions of the constant astigmatism equation as well as interpreting them as 
plastic flows on a sphere is merely a matter of routine.
Since computations quickly leave the realm of elementary functions,
the examples and illustrations scattered throughout this paper are only the
simplest ones. More are to follow in a subsequent paper.

\ack

The first named author was supported by Specific Research grant SGS/01/2011 of
the Silesian University in Opava. The second-named 
author was supported by GA\v{C}R (project P201/11/0356) and
by RVO institutional funding (I\v{C} 47813059).
We thank E.~Ferapontov, J.~Kluso\v{n}, I.S.~Krasil'shchik, and S.I.~Senashov for advice and valuable discussions. Section~\ref{sect:BTSP} was written at the kind suggestion of an anonymous referee of our previous work.

\end{document}